%% file: TimedExpressV31.tex
\documentclass[fleqn]{llncs}

\usepackage{pslatex,amssymb,pgf,tikz,lscape,cite}
\usetikzlibrary{arrows,automata}

\input{macros}

\begin{document}

\title{On Expressive Powers of Timed Logics:\\ Comparing Boundedness, Non-punctuality, and Deterministic Freezing}
\author{Paritosh K.~Pandya and Simoni S.~Shah}
\institute{Tata Institute of Fundamental Research,
Colaba, Mumbai \texttt{400005}, India \\
\email{\{pandya,simoni\}@tcs.tifr.res.in}}
\pagestyle{empty}
\maketitle

\begin{abstract}
Timed temporal logics exhibit a bewildering diversity of operators and the resulting decidability and expressiveness properties also vary considerably. We study the expressive power of timed logics $\tptlus$ and $\mtlus$ as well as of their several fragments. Extending the LTL EF games of Etessami and Wilke, we define $MTL$ Ehrenfeucht-Fra\"{\i}ss{\'e}  games on a pair of timed words. Using the associated EF theorem, we show that, expressively, the timed logics 
$\bmtlus$,  $\mtlfp$   and    $\mitlus$ (respectively incorporating the restrictions of boundedness, unary modalities and non-punctuality), are all pairwise incomparable. As our first main result,  we show that \mtlus\/ is \emph{strictly contained} within the freeze logic \tptlus\/ for both weakly and strictly monotonic timed words, thereby extending the result of Bouyer \emph{et al} and completing the proof of the original conjecture of Alur and Henziger from 1990. We also relate the expressiveness of a recently proposed deterministic freeze logic $\ttlxy$ (with NP-complete satisfiability) to $MTL$. 
As our second main result, we show by an explicit reduction that $\ttlxy$ lies strictly within the unary, non-punctual logic $\mitlfp$.
This shows that deterministic freezing with punctuality is expressible in the non-punctual \mitlfp.
\end{abstract}

\section{Introduction}
Temporal logics are well established formalisms for specifying qualitative ordering constraints on the  sequence of observable events. 
Real-time temporal logics extend this vocabulary  with specification of quantitative timing constraints between these events.

There are two well-established species of timed logics with linear time. The logic $\tptlus$ makes use of freeze quantification together with untimed temporal modalities and explicit constraints on frozen time values; the logic $\mtlus$ uses time interval constrained modalities $\until_I$ and $\since_I$. For example,the \tptlus\/ formula $x.(a\until(b\land T-x<2))$ and the \mtlus\/ formula $a\until_{[0,2)}b$ both characterize the set of words that have a letter  $b$ with time stamp $<$ 2 where this $b$ is preceded only by a string of letters $a$. Timed logics may be defined over timed words (also called pointwise time models)  or over signals (also called continuous time models). Weak monotonicity (as against strict monotonicity) allows a sequence of events to occur at the same time point. In this paper we confine ourselves to finite timed words with both weakly and strictly monotonic time, but the results straightforwardly carry over to infinite words too. 

In their pioneering studies \cite{AH90,AH93,AH94}, Alur and Henzinger  investigated the expressiveness and decidability properties of  timed logics $\mtlus$ and $\tptlus$. They showed that $\mtlus$ can be easily translated into $\tptlus$. Further, they conjectured, giving an intuitive example, that $\tptlus$ is more expressive than $\mtlus$ (see \cite{AH93} section 4.3). Fifteen years later, in a seminal paper, Bouyer \emph{et al} \cite{BCM05} formally proved that the purely future time logic $\tptlu$ is strictly more expressive than $\mtlu$  and that $\mtlus$ is more expressive than $\mtlu$, for both pointwise and continuous time. In this paper, we complete the picture by proving the original conjecture of Alur and Henzinger for the full logic $\mtlus$  with both future and past over pointwise time. 

In their full generality, $\mtlus$ and $\tptlus$ are both undecidable even for finite timed words. Several restrictions 
have been proposed to get decidable sub-logics (see \cite{OW08} for a recent survey). Thus, Bouyer \emph{et al.} \cite{BMOW08} introduced $\bmtlus$ with ``bounded'' intervals and showed that its satisfiability is $\mathit{EXPSPACE}$-complete. Alur and Henzinger argued, using reversal bounded 2-way deterministic timed automata $\mathit{RB2DTA}$, that the
logic $\mitlus$ permitting only non-singular (or non-punctual) intervals was decidable with $\mathit{EXPSPACE}$ complexity \cite{AH92,AFH96}. Unary modalities have played a special role in untimed logics \cite{EVW}, and we also consider unary fragments $\mtlfp$ and $\tptlfp$ in our study. Further sub-classes  can be obtained by combining the restrictions of bounded or non singular intervals and unary modalities. 
Decidable fragments of $\tptlus$ are less studied but two such logics can be found in \cite{PW09,PS10}.

In this paper, we mainly compare the expressive powers of various real-time temporal logics.
As our main tool we define an \emph{$m$-round MTL EF game}   with  ``until'' and ``since'' moves on two given timed words. As usual, the EF theorem equates the inability of any 
$\mtlus$ formula with modal depth $m$ from distinguishing two timed words to the existence of a winning strategy for the duplicator in $m$-round games.  Our EF theorem is actually parametrized by a  permitted set of time intervals, and it can be used for proving the lack of expressiveness of various fragments of $\mtlus$.

Classically, the EF Theorem has been a useful tool for proving limitations in
expressive power of first-order logic \cite{IK89,WT93}.
In their well-known paper, Etessami and Wilke \cite{EW96} adapted this to the LTL EF games to show the existence of the ``until'' hierarchy  in LTL definable languages. 
Our $\mathit{MTL}$ EF theorem is a generalization of this
to the timed setting.  We find that the use of EF theorem often leads to simple game theoretic  proofs of seemingly difficult questions about expressiveness of timed logics. The paper contains several examples of such proofs.

Our main expressiveness results are as follows. We show these results for finite timed words with weakly and strictly monotonic time. However, we remark that these results straightforwardly carry over to infinite timed words.
\begin{itemize}
\item We show that logics  $\bmtlus$, $\mitlus$ and $\mtlfp$ are all pairwise
incomparable.  These results indicate that the restrictions of boundedness, non-punctuality, and unary modalities are all semantically ``orthogonal'' in context of $\mathit{MTL}$.
\oomit{
Of these, logics $\bmtlus$ and $\mitlus$ form the currently known decidability frontier for timed temporal logics\footnote{In this paper, we confine ourselves to logics which are boolean closed. Extensions merging $\bmtlus$ with $\mitlus$ have been formulated as Flat, Co-Flat MTL and shown to be decidable \cite{BMOW08} and so is positive TPTL \cite{PW09}. However, these logics are not boolean closed.} where as $\ttlxy$ is a decidable fragment of the freeze logic
\tptlfp.
}
\item As one of our main results,
we show that the unary and future fragment $\tptlf$ of the freeze logic $\tptlus$ is not expressively contained within $\mtlus$ for both strictly monotonic and weakly monotonic timed words. Thus, $\mtlus$ is a strict subset of $\tptlus$ for pointwise time, as originally 
conjectured by Alur and Henzinger almost 20 years ago \cite{AH90, AH93, BCM05}.
\item It is easy to show that for strictly monotonic timed words, logic $\tptlus$ can be translated to the unary fragment $\tptlfp$ and for expressiveness the two logics coincide. For weakly monotonic time, we show that $\mtlus$ and $\tptlfp$ are expressively incomparable.
\end{itemize}

In the second part of this paper, we explore the expressiveness of  a recently proposed ``deterministic'' and ``unary'' fragment of $\tptlfp$ called $\ttlxy$. This is an interesting 
logic with exact automaton characterization as partially ordered two way deterministic timed automata \cite{PS10}. Moreover, by exploiting the properties of these automata, 
the logic has been shown to have NP-complete satisfiability.  
The key feature of this logic is the ``unique parsing'' of each timed word
against a given formula. Our main results on the expressiveness of  $\ttlxy$  are as follows.
\begin{itemize}
\item By an explicit reduction, we show that $\ttlxy$ is contained within the unary and non-punctual logic  $\mitlfp$.
The containment holds in spite of the fact that $\ttlxy$ can have freeze quantification and  punctual constraints (albeit only occurring deterministically).
\item Using the unique parsability of $\ttlxy$, we show that neither $\mitlfp$  nor $\bmtlfp$ are
expressively contained within $\ttlxy$.
\end{itemize}
Thus, the full logic $\tptlus$ is more expressive than $\mtlus$.
But its unary fragment with deterministic freezing,  $\ttlxy$, lies strictly within the unary and 
non-punctual logic $\mitlfp$. In our recent work \cite{PS11}, we have also shown by explicit reduction that the bounded fragment $\bmitlus$ is strictly contained within $\ttlxy$.  
Figure \ref{fig:concl} provides a succinct pictorial representation of all the expressiveness results achieved.

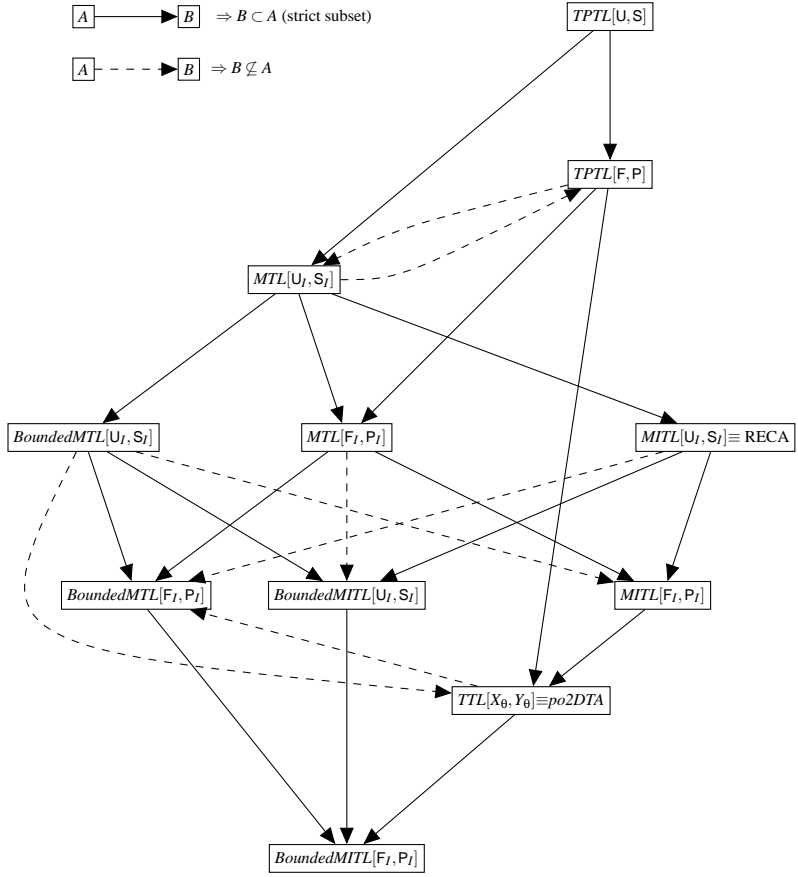
\begin{figure}[ht]
\begin{tikzpicture}[scale=0.7,transform shape]
\draw (6,18) node (A) [rectangle, draw] {\mtlus};
\draw (12,23) node (B) [rectangle,draw] {\tptlus};
\draw (12,20) node (C) [rectangle, draw] {\tptlfp};
\draw (7,15) node (F) [rectangle,draw] {\mtlfp};
\draw (2,15) node (G) [rectangle,draw] {\bmtlus};
\draw (14,15) node (H) [rectangle,draw] {\mitlus$\equiv$ RECA};
\draw (13,12) node (I) [rectangle,draw] {\mitlfp};
\draw (7,12) node (J) [rectangle,draw] {\bmitlus};
\draw (3,12) node (K) [rectangle, draw] {\bmtlfp};
\draw (10.5,10) node (E) [rectangle, draw] {\ttl$\equiv$\potdta};
\draw (7,7) node (L) [rectangle,draw] {\bmitlfp};

\draw (2,23) node (W) [rectangle, draw] {$A$};
\draw (4,23) node (X) [rectangle, draw] {$B$};
\draw (6,23) node{$\Rightarrow$  $B\subset A$ (strict subset)};
\draw (2,22) node (Y) [rectangle, draw] {$A$};
\draw (4,22) node (Z) [rectangle, draw] {$B$};
\draw (5,22) node{$\Rightarrow$  $B\not\subseteq A$};

\draw (W)[arrows= -triangle 45] -- (X);
\draw (Y)[arrows= -triangle 45,dashed] -- (Z);

\draw(B)[arrows= -triangle 45]-- (A);
\draw(B)[arrows= -triangle 45]-- (C);
\draw (C)[arrows= -triangle 45] -- (E);
\draw (E)[arrows= -triangle 45]-- (L);
\draw (A)[arrows= -triangle 45]-- (G);
\draw (A)[arrows= -triangle 45]-- (H);
\draw (A)[arrows= -triangle 45]-- (F);
\draw (H)[arrows= -triangle 45]-- (I);
\draw (F)[arrows= -triangle 45]-- (I);
\draw (C)[arrows= -triangle 45]-- (F);
\draw (G)[arrows= -triangle 45]-- (J);
\draw (H)[arrows= -triangle 45]-- (J);
\draw (G)[arrows= -triangle 45]-- (K);
\draw (F)[arrows= -triangle 45]-- (K);
\draw (J)[arrows= -triangle 45]-- (L);
\draw (K)[arrows= -triangle 45]-- (L);
\draw (I)[arrows= -triangle 45]-- (E);

\draw (E)[arrows= -triangle 45,dashed] --  (K);
\draw (G)[arrows= -triangle 45,dashed] -- (I);
\draw (F)[arrows= -triangle 45,dashed] --  (J);
\draw (H)[arrows= -triangle 45,dashed] -- (K);
\draw (G)[arrows= -triangle 45,dashed] .. controls (0,11) .. (E);
\draw (A)[arrows= -triangle 45,dashed] .. controls (8,18) .. (C);
\draw (C)[arrows= -triangle 45,dashed] .. controls (8,19) .. (A);

\end{tikzpicture}
\caption{Expressiveness of Timed Logics for Pointwise Time}
\label{fig:concl}
\end{figure}

The rest of the paper is organized as follows. Section \ref{sec:timedlogic} defines various timed logics. The $MTL$ EF games and the EF Theorem are given in Section \ref{sec:efgames}. Section \ref{sec:mtlsublogic} explores the relative expressiveness of various fragments of $\mtlus$ and
the subsequent section compares $\tptlus$ to $\mtlus$. Section \ref{sec:ttlmtl} studies the expressiveness of $\ttlxy$ relative to sub logics
of $\mtlus$.  

\section{Timed Temporal Logics: Syntax and Semantics}
\label{sec:timedlogic}
We provide a brief introduction of the logics whose expressiveness is investigated in this paper.
\subsection{Preliminaries}
Let  $\mathbb{R}, \mathbb{Z}$ and $\mathbb{N}$ be the set of reals, rationals, integers, and natural numbers, respectively and $\rlpos$ be the set of non-negative reals. An \textit{interval} is a convex subset of $\rlpos$, bounded by non-negative integer constants or $\infty$. The left and right ends of an interval may be open ( "(" or ")" ) or closed ( "[" or "]" ). We denote by $ \langle x,y \rangle$ a generic interval whose ends may be open or closed.  
An interval is said to be \textit{bounded} if it does not extend to infinity. It is said to be \textit{singular} if it is of the form $[c,c]$ for some constant $c$, and non-singular (or non-punctual) otherwise. 
We denote by $\zintv$ all the intervals (including singular
intervals $[c,c]$ and unbounded intervals $[c,\infty)$), by $\extintv$ the set of all non-punctual (or extended) intervals, and by $\bintv$ the set of all bounded intervals.
Given an alphabet $\Sigma$, its elements are used also as atomic propositions in logic, i.e.
the set of atomic propositions $\ap = \Sigma$.

A finite timed word is a finite sequence  $\rho = (\sigma_1,\tau_1), (\sigma_2,\tau_2), \cdots ,(\sigma_n,\tau_n)$, of event-time stamp pairs 
such that the 
sequence of time stamps is non-decreasing: $\forall i<n ~.~ \tau_i\leq \tau_{i+1}$. This gives weakly monotonic timed words. If
time stamps are strictly increasing, i.e. $\forall i<n ~.~ \tau_i< \tau_{i+1}$, the word is strictly monotonic. The length of $\rho$ is denoted by $\#\rho$, and $dom(\rho) = \{1,...\#\rho\}$. For convenience, we assume  that $\tau_1 = 0$ as this simplifies the treatment of ``freeze'' logics. The timed word $\rho$ can alternately be represented as 
$\rho=(\overline{\sigma},\overline{\tau})$ with $\overline{\sigma} = \sigma_1, \cdots, \sigma_n$ and 
$\overline{\tau} = \tau_1, \cdots,\tau_n$. Let $untime(\rho)=\overline{\sigma}$. We shall use the two representations interchangeably.
Let $T\Sigma^*$ be the set of timed words over the alphabet $\Sigma$.

\subsection{Metric Temporal Logics}  
The logic MTL extends Linear Temporal Logic by adding timing constraints to the "Until" and "Since" modalities of LTL. We parametrize this logic by a 
permitted set of intervals $\INTV$  and denote the resulting logic as
$\INTV\mtlus$. Let $\phi$ range over $\INTV\mtlus$ formulas, $a \in \Sigma$ and
$I \in \INTV$.  The syntax of $\INTV\mtlus$ is as follows:
\[
\phi ::= a ~\mid~ \phi \land\phi ~\mid~ \neg \phi ~\mid~ \phi\until_I\phi ~\mid~ \phi \since_I\phi
\]
Let $\rho=(\overline{\sigma},\overline{\tau})$ be a timed word and let $i\in dom(\rho)$. The semantics of \mtlus\/ formulas is as below:
\[
\begin{array}{rcl}
\rho,i\models a & \fif & \sigma_i=a\\
\rho,i\models \neg \phi & \fif & \rho,i\not\models \phi\\
\rho,i\models \phi_1\lor \phi_2 & \fif & \rho,i\models\phi_1 ~\mbox{or} ~\rho,i\models\phi_2\\
\rho,i\models \phi_1 \until_I \phi_2 & \fif & \exists j>i .~ \rho,j\models\phi_2 
~\mbox{and}~  \tau_j-\tau_i\in I\\
       &  &    ~\mbox{and}~ \forall i<k<j.~ \rho,k \models \phi_1 \\
\rho,i\models \phi_1 \since_I \phi_2 & \fif & \exists j<i ~.~ \rho,j\models\phi_2 
~\mbox{and}~ \tau_i-\tau_j\in I  \\
       &  &  ~\mbox{and}~ \forall j<k<i.~ \rho,k \models \phi_1 
\end{array}
\]
The language of an \INTV\mtlus\/ formula $\phi$ is given by $\mathcal L(\phi) = 
\{\rho ~\mid~ \rho, 1 \models \phi\}$.
Note that we use the "strict" semantics of $\until_I$ and $\since_I$ modalities.
We can define unary \textit{"future"} and \textit{"past"} modalities as: $\fut_I\phi := \top \until_I\phi$ and $\past_I\phi := \top \since_I \phi$. The subset of $\INTV\mtlus$ using only
these modalities is called $\INTV \mtlfp$. We can now define various well known variants of $MTL$.

\begin{itemize}
\item Metric Temporal Logic \cite{AH90,AH93}, denoted  $\mtlus = \zintv\mtlus$. This is obtained by choosing the set of intervals $\INTV = \zintv$. 
\item Unary $\mathit{MTL}$, denoted  \mtlfp = $\zintv\mtlfp$ uses only unary modalities. It is a
timed extension of the untimed unary temporal logic $UTL$ studied by \cite{EVW}.
\oomit{
\cite{EVW} defined a Unary temporal logic (UTL) as a sub logic of LTL with unary until and since modalities. We extend this to time and define a sub logic of $\mathit{MTL}$.
} 

\item Metric Interval Temporal Logic \cite{AFH96}, denoted $\mitlus = \extintv\mtlus$. In this logic, the timing constraints in the formulas are restricted to non-punctual (non-singular) intervals. \mitlfp\/ is \mitlus\/ confined to the unary modalities $F_I$ and $P_I$.

\item Bounded $\mathit{MTL}$\cite{BMOW08}, denoted $\bmtlus = \bintv\mtlus$. Other logics can be obtained as intersections of the above logics. Specifically, the logics  \bmtlfp, \bmitlus\/, and \bmitlfp\/ are defined  respectively as $\bintv\mtlfp$, $\bextintv\mtlus$, and $\bextintv\mtlfp$.

\item Let $\zintv^k$ denote the set of all intervals of the form $\langle i,j \rangle$ or 
$\langle  i,\infty)$, with $i,j \leq k$. Let $\bintv^k$ denote the set of all bounded 
(i.e. non-infinite) $\zintv^k$  intervals.
Then $\mtlus^k$ and $\bmtlus^k$ are respectively the logic  $\zintv^k\mtlus\/$ and
$\bintv^k\mtlus$. Also, given an  \mtlus\/ formula $\phi$, let $\maxint(\phi)$ denote the maximum integer constant (apart from $\infty$) appearing in its interval constraints.
\end{itemize}

\subsection{Freeze Logics}
These logics specify timing constraints by conditions on special variables, called freeze variables which memorize the time stamp at which a subformula is evaluated. Let $\mathcal X$ be a finite set of freeze variables. Let $x\in \mathcal X$ and let 
$\nu: \mathcal X\to \rlpos$ be a valuation which assigns a non-negative real number to each freeze variable. Let $\nu_0$ be the initial valuation such that
$\forall x ~.~ \nu_0(x) = 0$ and let 
$\nu (x\leftarrow r)$ denote the valuation such that $\nu(x \leftarrow r)(x) = r$ and $\nu(x \leftarrow r)(y)= \nu(y)$ if $x \not= y$.\\
A \emph{timing constraint} $g$ in freeze logics has the form:\\
\hspace*{1cm} $ g ~:=~ g_1 \land g_2 ~\mid~ x-T\approx c $ where $\approx\in \{<,\leq, >,\geq,=\}$ and $c\in\mathbb{Z}$.\\
Let $\nu,t \models g$ denote that the timing constraint $g$ evaluates to $\mathit{true}$ in valuation $\nu$ with $t \in \rlpos$ assigned to the variable $T$.

\paragraph{\tptlus} given by \cite{AH94,R99}, is an extension of LTL with freeze variables.
Let $g$ be a guard as defined above. The syntax of a $\tptlus$ formula $\phi$ is  as follows:
\[
 \phi := a ~\mid~ g ~\mid~  \phi\until\phi ~\mid~ \phi\since\phi ~\mid~ x. \phi 
 ~\mid~ \phi \lor\phi ~\mid~ \neg\phi
\]
The semantics of \tptlus\/ formulas over a timed word $\rho$ with 
$i\in dom(\rho)$ and valuation $\nu$ is as follows. The boolean connectives have their usual meaning.\\
\hspace*{1cm}
\begin{tabular}{rcl}
$\rho,i,\nu\models a$ & $\fif$ & $\sigma_i=a$\\
$\rho,i,\nu\models \phi_1\until\phi_2$ & $\fif$ & $\exists j>i ~.~ \rho,j,\nu\models\phi_2
 ~\mbox{and}~ 
     \forall i<k<j ~.~ \rho,k,\nu\models\phi_1$\\
$\rho,i,\nu\models \phi_1\since\phi_2$ & $\fif$ & $\exists j<i ~.~ \rho,j,\nu\models\phi_2 
~\mbox{and}~ 
     \forall j<k<i ~.~ \rho,k,\nu\models\phi_1 $\\
$\rho,i,\nu\models x.\phi$ & $\fif$ & $\rho,i,\nu(x\to\tau_i)\models \phi$\\
$\rho,i,\nu\models g$ & $\fif$ & $\nu,\tau_i \models g$ \\
\end{tabular} \\
The language defined by a \tptlus\/ formula $\phi$ is given by $L(\phi) = \{\rho ~\mid~ \rho,1,\nu_0\models\phi\}$. Also, \tptlfp\/ is the unary sub logic of \tptlus.

\paragraph{Deterministic Freeze Logic} \ttl\/ is a sub logic of \tptlus.
A \textit{guarded event} over an alphabet $\Sigma$ and a finite set of freeze variables $\mathcal X$ is a pair $\theta = (a,g)$ where $a \in \Sigma$ is an event and $g$ is a timing constraint over $\mathcal X$ as defined before.
Logic \ttl\/ uses the deterministic modalities $X_\theta$ and $Y_\theta$ which access the position with the \textit{next} and \textit{previous} occurrence of a guarded event, respectively. This is the timed extension of logic $TL[X_a,Y_a]$ \cite{DKL10} using freeze quantification. The syntax of a \ttl\/ formula $\phi$ is as follows: \\
$
\hspace*{1cm}
\phi := \top ~\mid~ \theta ~\mid~ SP\phi ~\mid~EP \phi ~\mid~ X_\theta \phi ~\mid~ Y_\theta \phi ~\mid~ x.\phi ~\mid~ \phi\lor\phi ~\mid~ \neg\phi
$ \\
The semantics of \ttl\/ formulas over timed words is as given below. $\top$ denotes the formula $\mathit{true}$. This and the boolean operators have their usual meaning. 
\[
\begin{array}{rcl}
\rho,i,\nu\models \theta & \fif & \sigma_i = a ~\mbox{and}~ \nu,\tau_i\models g ~~\mbox{where}~ \theta=(a,g) \\
\rho,i,\nu\models SP \phi & \fif & \rho,1,\nu\models \phi\\
\rho,i,\nu\models EP \phi & \fif & \rho,\#\rho,\nu\models \phi\\
\rho,i,\nu\models X_\theta \phi & \fif & \exists j>i ~.~ \rho,j,\nu\models \theta 
    ~\mbox{and}~
    \forall i<k<j . \\
   & & \quad \rho,k,\nu \not\models\theta    ~\mbox{and}~ \rho,j,\nu\models  \phi\\
\rho,i,\nu\models Y_\theta \phi & \fif & \exists j<i ~.~ \rho,j,\nu\models \theta   ~\mbox{and}~ \forall j<k<i . \\
   & & \quad \rho,k,\nu \not\models\theta   ~\mbox{and}~ \rho,j,\nu\models \phi\\
\rho,i,\nu\models x.\phi & \fif & \rho,i,\nu(x\leftarrow \tau_i)\models \phi
\end{array}
\]



\section{EF Games for \INTV\mtlus}
\label{sec:efgames}
We extend the LTL EF games of \cite{EW96} to timed logics,
and use these to compare expressiveness of various instances of the
generic logic $\INTV\mtlus$.
\oomit{
Recall that we have defined a generic logic $\INTV\mtlus$, and, by restricting the set of permitted intervals $\INTV$ we obtained various sub logics. 
The EF games here are also parametrized by the set of intervals $\INTV$.
Yet another restriction is to use the unary operators $F_I$ and $P_I$ in place of $\until_I$ and $\since_I$ giving logics $\INTV \mtlfp$. The EF games are adapted to this restriction too.
}
Let $\INTV$ be a given set of intervals. A $k$-round \INTV\mtlus-EF game is played between 
two players, called \ssp\/ and \ddp, on a pair of timed words 
$\rho_0$ and $\rho_1$.  
A configuration of the game (after any number of rounds) is a pair of positions $(i_0,i_1)$ with $i_0 \in dom(\rho_0)$ and $i_1 \in dom(\rho_1)$. A configuration is called partially isomorphic, denoted $isop(i_0,i_1)$ iff $\sigma_{i_0}=\sigma_{i_1}$. 

The game is defined inductively on $k$ from a starting configuration $(i_0,i_1)$ and results in either the \ssp\/ or \ddp\/ winning the game.
The \ddp\/ wins the $0$-round game  iff  $isop(i_0,i_1)$. The $k+1$ round game is played by first playing one round from the starting position. Either the spoiler wins in this round (and the game is terminated) or the game results into a new configuration $(i_0',i_1')$. The game then proceeds inductively with $k$-round play from the configuration $(i_0',i_1')$. The \ddp\/ wins the game only if it wins every round of the game.
We now describe one round of play from a starting configuration $(i_0,i_1)$.
\begin{itemize}
 \item At the start of the round, if $\neg isop(i_0,i_1)$ then the \ssp\/ wins the game and the game is terminated. Otherwise,
 \item The \ssp\/ chooses one of the words by choosing $\delta \in \{0,1\}$. Then $\deltabar=(1-\delta)$ gives the other word.
 The \ssp\/ also chooses either an $\until_I$-move or a $\since_I$ move, including an interval $I \in \INTV$.
 The remaining round is played in two parts.
\end{itemize}
\noindent{\em $\until_I$ Move}
\begin{itemize}
\item \textit{Part I:} The \ssp\/ chooses a position $i'_\delta$ such that $i_\delta < i'_\delta \leq \#\rho_\delta$ and $(\tau_\delta[i'_\delta] - \tau_\delta[i_\delta]) \in I$.
\item The \ddp\/ responds\footnote{The \ddp\/ can make use of the knowledge of $I$ to choose his move. This is needed as illustrated in the proof of Theorem \ref{theo:mtlfrag2}.}
by choosing a position $i'_\deltabar$ in the other word s.t. ${i}_\deltabar < i'_\deltabar \leq \#\rho_\deltabar$ and $(\tau_\deltabar[i'_\deltabar] - \tau_\deltabar[i_\deltabar]) \in I$. 
If the \ddp\/ cannot find such a position, the \ssp\/ wins the game. Otherwise the play continues to Part II.
\item \textit{Part II}: \ssp\/ chooses to play either $F$-part or $U$-part.
\begin{itemize}
\item $F$-part: the round ends with configuration $(i'_0,i'_1)$.
\item $U$-part: \ssp\/ verifies that $i'_{\delta}-i_\delta=1$ iff $i'_{\deltabar}-i_{\deltabar}=1$
and \ssp\/ wins the game if this does not hold. Otherwise \ssp\/ checks whether $i'_{\delta}-i_\delta=1$. If \emph{yes}, the round ends with configuration $(i'_0,i'_1)$. If \emph{no}, 
\ssp\/ chooses a position $i''_\deltabar$ in the other word such that 
$i_\deltabar < i''_\deltabar < i'_\deltabar$. The \ddp\/ responds by choosing
$i''_\delta$ such that $i_\delta < i''_\delta < i'_\delta$. 
The round ends with the configuration $(i''_0,i''_1)$.
\end{itemize}
\end{itemize}
\noindent{\em $\since_I$ Move} This move is symmetric to $\until_I$ where the \ssp\/ chooses positions $i'_\delta$ as well as $i''_\deltabar$ in ``past'' 
and the \ddp\/ also responds accordingly. In Part II, the \ssp\/ will a have choice of 
$P$-part or $S$-part. We omit the details. This completes the description of the game.

\begin{definition}
Given two timed words $\rho_0,\rho_1$ and $i_0\in dom(\rho_0), i_1\in dom(\rho_1)$, we define 
\begin{itemize}
\item $(\rho_0,i_0) \gameeq{k}^{\INTV} (\rho_1,i_1)$ iff for every $k$-round \INTV \mtlus\/  EF-game over the words $\rho_0,\rho_1$ and starting from the configuration $(i_0,i_1)$, the \ddp\/ always has a winning strategy.
\item $(\rho_0,i_0) \formeq{k}^{\INTV} (\rho_1,i_1)$ iff for every \INTV\mtlus\/ formula $\phi$ of operator depth $\leq k$, $\rho_0,i_0 \models \phi \Leftrightarrow \rho_1,i_1\models \phi$.
\qed
\end{itemize}
\end{definition}
We shall now state the \INTV \mtlus\/ EF theorem. Its proof is a straight-forward extension of the proof of LTL EF theorem of \cite{EW96}. The only point of interest is that there is
no a priori bound on the set of intervals that a modal depth $n$ formula can use and hence the set of isomorphism types seems potentially infinite. However, given timed words $\rho_0$ and $\rho_1$, we can always restrict these intervals to not go beyond a constant 
$k$  where $k$ is the smallest integer larger than the biggest time stamps in $\rho_0$ and $\rho_1$. This restricts the isomorphism types to a finite cardinality.
The complete proof is given in detail in Appendix  \ref{app:efthmproof}.
\begin{theorem}
\label{thm:ef}
 $(\rho_0,i_0) \gameeq{k}^{\INTV} (\rho_1,i_1)$ if and only if $(\rho_0,i_0) \formeq{k}^{\INTV} (\rho_1,i_1)$ \qed
\end{theorem}
When clear from context, we shall abbreviate $\gameeq{k}^{\INTV}$ by $\gameeq{k}$ and
$\formeq{k}^{\INTV}$ by $\formeq{\INTV}$.
As temporal logic formulas are anchored to initial position $1$, define
$\rho_0 \formeq{k} \rho_1 \iff (\rho_0,1) \formeq{k} (\rho_1,1)$ and 
$\rho_0 \gameeq{k} \rho_1 \iff (\rho_0,1) \gameeq{k} (\rho_1,1)$. It follows from the EF Theorem that $\rho_0 \formeq{k} \rho_1$ if and only if $\rho_0 \gameeq{k} \rho_1$.

We can modify the $\INTV \mtlus$ EF game to match the sub logic $\INTV \mtlfp$.
An $\INTV \mtlfp$ game is obtained by the restricting $\INTV \mtlus$ game such that in PART II
of any round, the \ssp\/ always chooses an $F$-part or a $P$-part.
The corresponding $\INTV \mtlfp$ EF Theorem also holds.


\section{Separating sub logics of \mtlus\/}
\label{sec:mtlsublogic}
Each formula of a timed logic defines a timed language.
Let $\mathcal L(\mathcal G)$ denote the set of languages definable by the formulas of logic $\mathcal G$. A logic $\mathcal G_1$ is  at least as expressive as (or contains) logic $\mathcal G_2$ if $\mathcal L(G_2) \subseteq \mathcal L(G_1)$. 
This is written as $\mathcal G_2 \subseteq \mathcal G_1$. Similarly, we can define 
$\mathcal G_2 \subsetneq \mathcal G_1$ (strictly contained within), $\mathcal G_2 \not \subseteq \mathcal G_1$ (not contained within), $\mathcal G_2 \# \mathcal G_1$ (incomparable), and
$\mathcal G_2 \equiv \mathcal G_1$ (equally expressive).

We consider three sub logics of \mtlus\/ namely \mtlfp, \mitlus\/ and \bmtlus.These have fundamentally different restrictions and using their corresponding EF-games, we show that they are all incomparable with each other.\footnote{It was already observed by Bouyer \emph{et al} \cite{BMOW08} that \bmtlus\/ and \mitlus\/ have separate expressiveness.} 
\oomit{
In fact, we examine the sub-logics \bmitlus, \bmtlfp\/ and \mitlfp\/ and show that they are incomparable with \mtlfp, \mitlus\/ and \bmtlus, respectively. We only state the theorems here. Their game-based proofs may be found in Appendix B.
}

\begin{theorem} \label{theo:mtlfrag1}
$\mitlfp \nsubseteq \bmtlus$
\end{theorem}
\begin{proof}
Consider the \mitlfp\/ formula $\phi:= \fut_{[0,\infty)}(a\land \fut_{(1,2)} c)$.
Consider a family of words $\aaa_n$ and $\bbb_n$.
We have $untime(\aaa_n) = untime(\bbb_n) = a^{n+1}c$ with the $a$'s occurring at integral time stamps $0,1, \ldots, n$ in both words. In $\aaa_n$, the letter $c$ occurs at time $n+2.5$ and hence time distance between any $a$ and $c$ is more than $2$. In $\bbb_n$, the $c$ occurs at time $n+1.5$ and the time distance between the $c$ and the preceding $a$  is in $(1,2)$. Clearly,  $\aaa_n \not \models \phi$ whereas $\bbb_n \models \phi$ for any $n>0$.

We prove the theorem using an $m$-round $\bintv^k\mtlus$ EF game on the words $\aaa_n$ and $\bbb_n$ where  $n = mk$. We show that \ddp\/ has a winning strategy.
Note that in such a game the \ssp\/ is allowed to choose intervals at every round with maximum upper bound of $k$ and hence can shift the pebble at most $k$ positions to the right. It is easy to see that the \ssp\/ is never able to place a pebble on the last $c$. Hence, the \ddp\/ has a winning strategy 
where she exactly copies the \ssp\/ moves. Using the EF theorem,
we conclude that no modal depth $n$ formula of logic $\bintv^k\mtlus$ can separate the words
$\aaa_n$ and $\bbb_n$. Hence,  there doesn't exist a \bmtlus\/ formula giving the language  $L(\phi)$. 
\oomit{
\begin{figure}
\begin{tikzpicture}
\draw (2,2) node{$\mathcal A_n$};
\draw(3,2.5) node{};\draw (4,2.5) node{a};\draw (5,2.5) node{a};\draw (7,2.5) node{a}; \draw (8,2.5) node{a}; \draw (10.5,2.5) node{b}; \draw (11,2.5) node{b};\draw (12,2.5) node{b};
\draw(3,2) node{I}--(4,2) node{I}--(5,2) node{I};
\draw (5,2)[dashed]--(7,2) node{I};
\draw (7,2)--(8,2) node{I}--(9,2) node{I}--(10,2) node{I}--(10.5,2) node{l}-- (11,2) node{I}--(12,2) node{I};
\draw (12,2)[dashed]-- (13,2) node{};
\draw(3,1.5) node{0};\draw (4,1.5) node{1};\draw (5,1.5) node{2};\draw (7,1.5) node{$n$}; \draw (8,1.5) node{$n+1$}; \draw (9,1.5) node{$n+2$};\draw (10,1.5) node{$n+3$}; \draw (11,1.5) node{$n+4$};\draw (12,1.5) node{$n+5$};

\draw (2,0.5) node{$\mathcal A_n$};
\draw(3,1) node{};\draw (4,1) node{a};\draw (5,1) node{a};\draw (7,1) node{a}; \draw (8,1) node{a};  \draw (11,1) node{b};\draw (12,1) node{b};
\draw(3,0.5) node{I}--(4,0.5) node{I}--(5,0.5) node{I};
\draw (5,0.5)[dashed]--(7,0.5) node{I};
\draw (7,0.5)--(8,0.5) node{I}--(9,0.5) node{I}--(10,0.5) node{I}-- (11,0.5) node{I}--(12,0.5) node{I};
\draw (12,0.5)[dashed]-- (13,0.5) node{};
\draw(3,0) node{0};\draw (4,0) node{1};\draw (5,0) node{2};\draw (7,0) node{$n$}; \draw (8,0) node{$n+1$}; \draw (9,0) node{$n+2$};\draw (10,0) node{$n+3$}; \draw (11,0) node{$n+4$};\draw (12,0) node{$n+5$};
\end{tikzpicture}
\caption{$\mitlfp \nsubseteq \bmtlus$}
\label{mitlfp:bmtlus}
\end{figure}
}
\qed
\end{proof}

\begin{theorem} \label{theo:mtlfrag2}
$\bmtlfp \nsubseteq \mitlus$
\end{theorem}
\begin{proof}
Consider the \bmtlfp\/ formula $\phi:= \fut_{(0,1)}(a\land \fut_{[3,3]} c)$. Consider a family of words $A_n$ such that $untime(A_n) = a^{2n+1}c^{2n+1}$. Let $\delta=1/(2n+2)^2$ and $\epsilon=1/(2n+2)^4$. All the $a$'s are in the interval (0,1) at time stamps $i\delta$ and all the $c$'s are in the interval $(3,4)$, at time stamps $3+i\delta+\epsilon$ for $1\leq i\leq 2n+1$. Every $a$ has a paired $c$, which is at a distance $3+\epsilon$ from it. Hence, $\forall n ~.~ A_n\not \models \phi$. Let $B_n$ be a word identical 
to $A_n$ but with the middle $c$ shifted leftwards by $\epsilon$, so that it is exactly at a distance of $3$ t.u. (time units) from the middle $a$. Thus, $B_n \models \phi$.

We prove the theorem using the $n$-round $\extintv$\mtlus EF game on the words $A_{2n}$ and $\bbb_{2n}$
where we can show that \ddp\/ has a winning strategy. 
This proves that no modal depth $n$ formula of logic  $\mitlus$ can separate $A_{2n}$ and $B_{2n}$. 
Hence, there is no $\mitlus$ formula giving $\mathcal L(\phi)$
The full description of the \ddp\/ strategy can be found in the Appendix \ref{app:mtlus}.
\qed
\end{proof}

\begin{theorem} \label{theo:mtlfrag3}
 \begin{itemize}
  \item $\bmtlus \not \subseteq \mtlfp$ over strict monotonic timed words (and hence also over weakly monotonic timed words).
  \item $\bmtlus \not \subseteq \tptlfp$ over weakly monotonic timed words. \qed 
 \end{itemize}
\end{theorem}
These results follow by embedding untimed LTL into logics MTL as well as TPTL. The proof can be found in Appendix \ref{app:mtlus}.

\section{TPTL and MTL}
\label{sec:tptlmtl}

Consider the $\tptlf$ formula $\phi_1 \df x. \fut(b \land \fut (c \land T-x \leq 2))$. 
Bouyer \emph{et al} \cite{BCM05} showed that this formula cannot be expressed in $\mtlu$  for pointwise
models. They also gave an $\mtlus$ formula equivalent to it thereby showing that $\mtlus$ is strictly more expressive than $\mtlu$. Prior to this, 
Alur and Henzinger \cite{AH93} considered the formula $\Box(a \Rightarrow \phi_1)$ and they conjectured that this cannot be expressed within $\mtlus$. Using a variant of this formula
and the $\mtlus$ EF games, we now show that $\tptlf$ is indeed expressively incomparable with $\mtlus$.

In Theorem \ref{theo:mtlfrag3} we showed that $\bmtlus \not \subseteq \tptlfp$ over weakly monotonic timed words. We now consider the converse.
\begin{theorem}
\label{theo:main}
\mbox{$\mathit{TPTL[F]}$} $\not\subseteq$ \mtlus\/ over strictly monotonic timed words (and hence also for weakly monotonic timed words).
\end{theorem}
\begin{proof}
Let the $\mathit{TPTL[F]}$ formula
$\phi := \fut p.[a\land \{\fut (b\land (T-p\in(1,2))\land \fut(c\land (T-p\in (1,2))))\}]$. This formula characterizes the set of timed words which have an $a$ followed by a $b$ and then a $c$ such that the time lag between the $a$ and $b$ is in the interval $(1,2)$ and the time lag between the $a$ and $c$ is also in $(1,2)$. We show that there is no \mtlus\/ formula that expresses the language defined by $\phi$.

The idea behind the proof is the following. We will design two families of strictly monotonic timed words $\aaa_{n,k}$ and $\bbb_{n,k}$ ($n>0$), such that $\aaa_{n,k} \models \phi$ and 
$\bbb_{n,k} \not \models \phi$. We will then show that for $n$ round $\zintv^k\mtlus$ EF games over $\aaa_{n,k}$ and $\bbb_{n,k}$ the duplicator has a winning strategy. Hence, no $n$ modal depth  $\zintv^k\mtlus$ formula can distinguish words $\aaa_{n,k}$ and $\bbb_{n,k}$. Thus, there is no formula in $\zintv\mtlus$ giving $L(\phi)$.
\end{proof}

\paragraph{Designing the words}
Fix some $n,k$. Let $m=2n(k+1)+1$, $\delta=1/2m$ and $\epsilon<<\delta$. First, we shall describe $\bbb_{n,k}$. The first event
is an $a$ at time stamp 0. (This event is included since  all words must begin with time stamp 0.)
Following this, there are no events in the interval $(0,k]$. From $k+1$ onwards, it has $m$ copies of identical and overlapping segments of length $2+\epsilon$ time units each. If the $i^{th}$ segment $seg_i$ begins at some time stamp (say $t$) then $seg_{i+1}$ begins at $(t+1-\delta)$. The beginning of each segment is marked by an $a$ at $t$, followed by a $b$ in the interval $(t+2-2\delta+2\epsilon,t+2-\delta-2\epsilon)$, and a $c$ in the interval $(t+2, t+2+\epsilon)$, as shown in figure \ref{fig:gameseg}. Note that all the events must be placed such that no two events are exactly at an integral distance from each other (this is possible, since $n$ and $k$ are finite and time is dense). Let $X=n(k+1)+1$. The $X^{th}$ segment is the middle segment, which is padded by $n(k+1)$ segments on either side. Let $seg_X$ begin at time stamp $x$ and the following segments begin at $y$ and $z$ respectively, as shown in the figure \ref{fig:tptlefgame}.  Let $p_x$ denote the position corresponding to the time stamp $x$ in both words.\\
$\aaa_{n,k}$ is identical to $\bbb_{n,k}$ except for the $X^{th}$ segment where the corresponding $c$ is shifted leftwards to be in the interval $(x+2-\epsilon,x+2)$. Let $p^{A}$ and $p^{A'}$ denote the positions of $c$ corresponding to $seg_{X}$ and $seg_{X+1}$ in $\aaa_{n,k}$ respectively. Similarly, let $p^{B}$ and $p^{B'}$ denote the positions of $c$ corresponding to $seg_{X}$ and $seg_{X-1}$ in $\bbb_{n,k}$ respectively. 

Note that $\bbb_{n,k}$ is such that for every $a$, there exists a $c$ at a distance $(1,2)$ from it, but the $b$ between them is at a distance $<1$ t.u. from the $a$. In addition, every $a$ has a $b$ at a distance $(1,2)$ from it, but the subsequent $c$ is at a distance $>2$ t.u. from the $a$.
See Figure \ref{fig:tptlefgame}.
Hence, $\forall n,k>0$, $\bbb_{n,k}\not\models\phi$. On the other hand, $\aaa_{n,k}$ is identical to $\bbb_{n,k}$ except for the $(n(k+1)+1)^{st}$ segment for which the $c$ is shifted left so that
$a$ has a $b$ followed by $c$, both of which are within time distance $(1,2)$ from the $a$. Hence, $\forall n,k>0$, $\aaa_{n,k}\models \phi$. Since all the events occur at time stamps $>k$, the \ssp\/ cannot differentiate between integer boundaries. This enables us to disregard the integer boundaries between the events through the play of the game. Moreover, since the words are such that no two events are exactly integral distance apart from each other, the \ssp\/ is forced to choose a non-singular interval in every round. 

\begin{figure}
\begin{tikzpicture}[scale=0.7,transform shape]
\draw (1,1) node{I} -- (9,1) node{I} -- (14,1) node{l}-- (14.4,1) node{l}-- (15.1,1) node{l}--(15.5,1) node{l}-- (17,1) node{I}-- (17.4,1) node{l}; 
\draw (1,0.3) node{$t$}; \draw (9,0.3) node{$t+1$}; \draw (17,0.3) node{$t+2$}; \draw (14.2,0.7) node{$2\epsilon$}; \draw (15.3,0.7) node{$2\epsilon$}; \draw (16.3,0) node{------$\delta$------}; \draw (14.7,0) node{------$\delta$-------}; \draw (17.2,0.7) node{$2\epsilon$};
\draw (17.2,1.3) node{$c$}; \draw (1,1.3) node{$a$}; \draw (14.7,1.3) node{$b$};
\draw [dotted] (14,1)-- (14,0); \draw [dotted] (15.5,1)-- (15.5,0); \draw [dotted] (17,1)-- (17,0); \draw [dotted] (14.4,1)-- (14.4,0.5); \draw [dotted] (15.1,1)-- (15.1,0.5);
\end{tikzpicture}
\caption{\mtlus\/ EF game : A single segment in $\bbb_{n,k}$}
\label{fig:gameseg}
\end{figure}

\begin{figure}
\begin{tikzpicture}[scale=0.7,transform shape]
\draw (0,7) node{$\mathcal A_{n,k}$};

\draw (8.9,8) node{$p^A$}; \draw (12.6,8) node{$p^{A'}$};
\draw (6.1,4.3) node{$p^{B'}$}; \draw (9.4,4.3) node{$p^B$}; 
\draw (1,7.5) node{$a$};\draw (1.4,7.5) node{$b$};  \draw (4.2,7.5) node{$a$};\draw (7.4,7.5) node{$a$}; \draw (2.9,7.5) node{$c$}; \draw (4.5,7.5) node{$b$}; \draw (6,7.5) node{$c$}; \draw (7.9,7.5) node{$b$}; \draw (8.9,7.5) node{$c$}; \draw (10.8,7.5) node{$a$}; \draw (11.1,7.5) node{$b$}; \draw (12.5,7.5) node{$c$}; \draw (14,7.5) node{$a$}; \draw (14.3,7.5) node{$b$}; 
\draw(1,6) node {$p_x$};
\draw(1,7) node{I}--(1.8,7) node{l}-- (2.6,7) node {l}-- (4.2,7) node{x}--(5,7) node{I}-- (5.8,7) node{l}-- (7.4,7) node{0}-- (8.2,7) node{x}-- (9,7) node{I}-- (10.8,7) node{l}-- (11.5,7) node{0}-- (12.2,7) node{x}-- (14,7) node{l}-- (14.8,7) node{l}-- (15.6,7) node{0};

\draw(1,6.5) node{$x$}; \draw(4.2,6.5) node{$y$}; \draw(5,6.5) node{$x+1$}; \draw(7.4,6.5) node{$z$}; \draw(8.2,6.5) node{$y+1$}; \draw(9,6.5) node{$x+2$}; \draw(11.5,6.5) node{$z+1$}; \draw(12.2,6.5) node{$y+2$}; \draw(15.6,6.5) node{$z+2$}; 

\draw (0,3.5) node{$\mathcal B_{n,k}$};

\draw (1,4) node{$a$};\draw (1.4,4) node{$b$};   \draw (4.2,4) node{$a$};\draw (7.4,4) node{$a$}; \draw (2.9,4) node{$c$}; \draw (4.5,4) node{$b$}; \draw (6,4) node{$c$}; \draw (7.9,4) node{$b$}; \draw (9.3,4) node{$c$}; \draw (10.8,4) node{$a$}; \draw (11.1,4) node{$b$}; \draw (12.5,4) node{$c$}; \draw (14,4) node{$a$}; \draw (14.3,4) node{$b$};

\draw(1,3.5) node{I}--(1.8,3.5) node{l}-- (2.6,3.5) node {l}-- (4.2,3.5) node{x}--(5,3.5) node{I}-- (5.8,3.5) node{l}-- (7.4,3.5) node{0}-- (8.2,3.5) node{x}-- (9,3.5) node{I}-- (10.8,3.5) node{l}-- (11.5,3.5) node{0}-- (12.2,3.5) node{x}-- (14,3.5) node{l}-- (14.8,3.5) node{l}-- (15.6,3.5) node{0};
\draw(1,2.5) node {$p_x$};
\draw(1,3) node{$x$}; \draw(4.2,3) node{$y$}; \draw(5,3) node{$x+1$}; \draw(7.4,3) node{$z$}; \draw(8.2,3) node{$y+1$}; \draw(9,3) node{$x+2$}; \draw(11.5,3) node{$z+1$}; \draw(12.2,3) node{$y+2$}; \draw(15.6,3) node{$z+2$}; 

\draw[dotted] (9.5,4.2)--(12.4,6);
\draw[->>,dotted] (12.4,6) -- (12.4,7.2);
\draw[dotted] (8.9,7.3)--(8.9,6);
\draw[->>,dotted] (8.9,6)--(6.2,4.2);
\end{tikzpicture}
\caption{\mtlus\/ EF game : Duplicator's Strategy}
\label{fig:tptlefgame}
\end{figure}
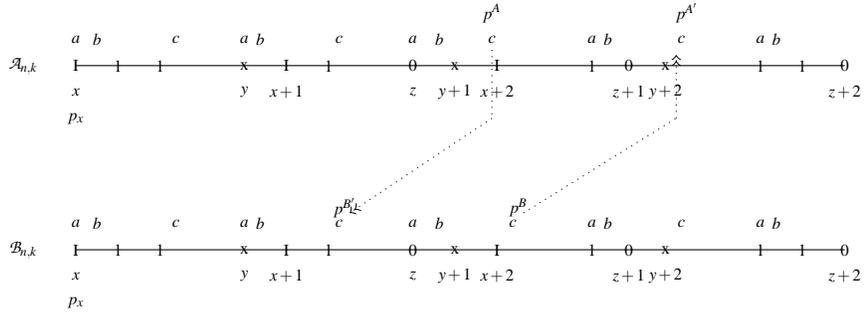

\oomit{
\paragraph{Choice of Intervals} The words $\aaa_{n,k}$ and $\bbb_{n,k}$ are such that no two events are at integral distance from each other. Hence, in every round of the $\zintv^k-\mtlus$ game, the \ssp\/ is forced to choose an interval of the form $(h,l)$ or $(l,\infty)$, where $h < l$ and $l\leq k$.
}

\paragraph{Key moves of \ddp} 
As the two words are identical except for the time stamp of the middle $c$, the strategy of \ddp\/ is to play a configuration of the form $(i,i)$ whenever possible. Such a configuration $(i,i)$ is called an identical configuration. The optimal strategy of \ssp\/ is to get out of identical configurations as quickly as possible. We give two example plays, where the \ssp\/ can force non-identical configuration (depicted by dotted arrows in figure \ref{fig:tptlefgame}). In first move, the \ssp\/ plays position $p_x$ which \ddp\/ duplicates giving the initial configuration of $(p_x,p_x)$.
\begin{enumerate}
 \item If the \ssp\/ chooses the interval $(1,2)$ and places its pebble at $p^A$, then the \ddp\/ will be forced to place its pebble at $p^{B'}$, which also occurs in the interval $x+(1,2)$. This is shown by downward dotted arrow in the figure.
 \item Alternatively, if the \ssp\/ chooses the interval $(2,3)$ and places a pebble at $p^B$ in $\bbb_{n,k}$, then the \ddp\/ is forced to place its pebble on $p^{A'}$, which is also in the interval $x+(2,3)$.
\end{enumerate}
In both cases, if $(i,j)$ is the resulting configuration, then $seg(i)-seg(j) = 1$. 

\paragraph{\ddp's copy-cat strategy}Consider the $p^{th}$ round of the game, with an initial configuration $(i_p,j_p)$. If the \ddp\/ plays in a manner such that the configuration for the next round is $(i_{p+1},j_{p+1})$ with $seg(i_p)-seg(i_{p+1}) = seg(j_p)-seg(j_{p+1})$, then it is said to have followed the \textit{copy-cat} strategy for the $p^{th}$ round.

\begin{proposition}\label{prop:ddpstrat}
The only case when the \ddp\/ can not follow the \textit{copy-cat} strategy in a round with initial configuration $(i,j)$, is when $i=j=p_x$ and the \ssp\/ chooses to first place its pebble on either $p^A$ or $p^B$ or when $i=p^A$ and $j=p^{B}$ and \ssp\/ chooses to place a pebble at $p_x$ in either word.
\end{proposition}
\begin{proof}
Firstly, note that $untime(\aaa_{n,k}) = untime (\bbb_{n,k})$ and the only position at which the two words differ is at $p^A$ (and correspondingly $p^B$), where $\tau_{p^B} - \tau_{p^A} < 2\epsilon$. By observing the construction of the words, we can infer that $\forall p \in dom (\aaa_{n,k})$, if $p\neq p_x$ then $\forall i\in\mathbb Z$ we have $\tau_p-\tau_{p^A}\in (i,i+1)$ iff $\tau_p-\tau_{p^B}\in (i,i+1)$. However, if the initial configuration is $(p_x,p_x)$ or $(p^A,p^B)$, then \ddp\/ may not be able to follow the \textit{copy-cat} strategy, since $p_A$ and $p_B$ lie on either side of $x+2$.
\qed
\end{proof}

The lemma below shows that in an $n$ round game, for each round, the \ddp\ can either achieve an identical configuration, or restrict  the segment difference between words to a maximum of $1$ in  which case  there are sufficient number of segments on either side for the \ddp\/ to be able to duplicate the \ssp's moves for the remaining rounds.
\begin{lemma}\label{lem:maingame}
For an $n$ round $\zintv^k$ \mtlus\/ EF game over the words $\aaa_{n,k},\bbb_{n,k}$ the \ddp\/ always has a winning strategy such that for any $1\leq p\leq n$, if $(i_p,j_p)$ is the initial configuration of the $p^{th}$ round then 
\begin{itemize}
\item $seg(i_p)-seg(j_p) ~\leq~ 1$  AND
\item If $seg(i_p)\neq seg(j_p)$ then
\begin{itemize}
\item[]$Min\{seg(i_p),seg(j_p)\} ~>~ (n-p+1)(k+1)$
\item[]$Max\{seg(i_p),seg(j_p)\} ~<~ m-(n-p+1)(k+1)$
\end{itemize}
\end{itemize}
\end{lemma}
\begin{proof}
The duplicator always follows \textit{copy-cat} strategy in any configuration whenever possible.
We can prove the lemma by induction on $p$. \\ 
\emph{Base step:} The lemma holds trivially for $p=1$, as starting configuration $(i_1,j_1) = (1,1)$. \\
\emph{Induction Step:} Assume that the lemma is true for some $p<n$. We shall prove that the lemma holds for $p+1$. Consider the $p^{th}$ round, with initial configuration $(i_p,j_p)$.\\
\emph{Case 1:} The \ddp\/ can follow \textit{copy-cat} strategy :\\
Then,  $seg(i_{p+1})-seg(j_{p+1}) = seg(i_p)-seg(j_p)$. By induction hypothesis, 
$seg(i_p)-seg(j_p)~\leq~ 1$ giving $seg(i_{p+1})-seg(j_{p+1}) ~\leq~ 1$. Also,
since exactly $k$ number of segments begin within a time span of $k$ time units,
if the \ssp\/ chooses an interval of the form $(h,l)$, with $l\leq k$, then we know that $seg(i_{p+1}) - seg(i_p) \leq k$ and the lemma will hold for $p+1$. If the \ssp\/ chooses an interval $(k,\infty)$ and places a pebble $k+1$ segments away in $\aaa_{n,k}$, the \ddp\/ also has to place its pebble at least $k+1$ segments away, thereby, either making $seg(i_{p+1})=seg(j_{p+1})$ or making $i_{p+1}$ and $j_{p+1}$ come closer to either end by at most $k+1$ segments.\\
\emph{Case 2:} The \ddp\/ can not follow \textit{copy-cat} strategy:\\
From proposition \ref{prop:ddpstrat}, this can happen only if $seg(i_p)=seg(j_p)= X$, the middle segment. In this case, we know that $seg(i_{p+1})-seg(j_{p+1})=1$, $X-2 \leq seg(i_{p+1})\leq X+2$ and $X-2\leq seg(j_{p+1}) \leq X+2$. Hence the lemma holds in this case too.
\qed
\end{proof}

\section{Comparing \ttlxy\/ with \mtlus\/ fragments}
\label{sec:ttlmtl}
\subsection{Embedding \ttl\/ into \mitlfp:} Fix a formula $\phi \in \ttlxy$. The formula $\phi$ may be represented by its parse tree $T_{\phi}$, such that the subformulas of $\phi$ form the subtrees of $T_{\phi}$. Let $\sub(n)$ denote the subformula corresponding to the subtree rooted at node $n$, and let $n$ be labelled by $\opr(n)$ which is the outermost  operator (such as $X_\theta, \lor,\neg, x.$ etc.) if $n$ is an interior node, and by the corresponding atomic proposition, if it is a leaf node. We will use the notion of subformulas and nodes interchangeably. The \textit{ancestry} of a subformula $n$ is the set of nodes in the path from the root up to (and including) $n$. \\
Let $\eta$ to range over subformulas of $\phi$ with $\eta_{root}$ denoting $\phi$. Logic $\ttlxy$ is a deterministic freeze logic. Hence, given a timed word $\rho$, in evaluating $\rho,1,\nu_0 \models \phi$, any subformula $\eta$ of $\phi$ needs to be evaluated only at a uniquely determined position in $dom(\rho) \cup \{\bot\}$ called $pos_\rho(\eta)$. We call this the \textit{Unique Parsability} property of \ttl\/ formulas. Here, notation $pos_\rho(\eta)= \bot$ indicates that such a position does not exist in $\rho$ and that the subformula $\eta$ plays no role in evaluating $\rho,1,\nu_0 \models \phi$. Also, $val_\rho(\eta)$ is the unique valuation function of freeze variables under which $\eta$ is evaluated. Note that $pos$ is strict w.r.t. $\bot$, i.e. if $\eta=OP(\ldots,\eta_1,\ldots)$ and $pos_\rho(\eta)=\bot$ then $pos_\rho(\eta_1)=\bot$. Also, $val$ is a partial function where $val_\rho(\eta)$ is defined only when $pos_\rho(\eta) \not= \bot$. We define $pos_\rho(\eta)$ together with $val_\rho(\eta)$ which are both simultaneously defined by induction on the depth of $\eta$. Firstly, define $pos_\rho(\eta_{root})=1$ and $val_\rho(\eta_{root})=\nu_0$. Now consider cases where $pos_\rho(\eta) = i ~(\not= \bot)$ and $val_\rho(\eta)=\nu$. 
\begin{itemize}
\item If $\eta= SP \eta_1$ then $pos_\rho(\eta_1) = 1$ and $val_\rho(\eta_1)=\nu$.
 \item  If $\eta= EP \eta_1$ then $pos_\rho(\eta_1) = \#\rho$ and $val_\rho(\eta_1)=\nu$.
\item If $\eta=\eta_1 \lor \eta_2$ or $\eta=\neg \eta_1$ then $pos_\rho(\eta_1) = pos_\rho(\eta_2)= i$ and $val_\rho(\eta_1) = val_\rho(\eta_2)= \nu$. 
\item If $\eta= x. \eta_1$ then $pos_\rho(\eta_1) = i$ and $val_\rho(\eta_1)= 
\nu(x \leftarrow \tau_i)$.
\item Let $\eta= X_\theta \eta_1$. Then, $pos_\rho(\eta_1) = \bot$ if 
$\forall k > i, ~\rho,k,\nu \not \models \theta$. Otherwise, $pos_\rho(\eta_1) = j$ s.t.
$j>i$ and $\rho,j,\nu\models \theta$ and $\forall i < k < j, ~\rho,k,\nu\not \models \theta$.
Moreover, $val_\rho(\eta_1) = \nu$.
\item The case of $\eta=Y_\theta \eta_1$ is symmetric to that of $\eta=X_\theta \eta_1$. 
\end{itemize}
Given a freeze variable $x$, let $anc_x(\eta)$ be the node in the ancestry of $\eta$ and nearest to it, at which $x$ is frozen. Hence, $anc_x(\eta)$ is the smallest ancestor $\eta'$ of $\eta$, which is of the form $x.\eta'$. If there is  no such ancestor, then let $anc_x(\eta)=\eta_{root}$. The following proposition follows from this definition.
\begin{proposition}\label{prop:anc}
 $val_\rho(\eta)(x) = \tau_{pos_\rho(anc_x(\eta))}$
\end{proposition}

\begin{lemma} \label{lem:ttltranspos} 
For any subformula $\eta$ of a \ttlxy\/ formula $\phi$,  we can effectively construct an 
\mitlfp\/ formula $\alpha(\eta)$ such that $\forall \rho\in T\Sigma^*$ we have $pos_\rho(\eta)=j$ iff $\rho,j \models \alpha(\eta)$.
\end{lemma}
\begin{proof}
The construction of $\alpha(\eta)$ follows the inductive definition of $pos_\rho(\eta)$, and the lemma may be proved by induction on the depth of $\eta$. Consider any timed word $\rho$.
\begin{itemize}
\item Firstly, $\alpha(\eta_{root}) = \neg P_{(0,\infty)} \top$. Therefore, $\rho,i\models \alpha(\eta_{root})$ iff $i=1$.
\item Similarly, if $\eta = SP \eta_1$ then $\alpha(\eta_1) = \neg P_{(0,\infty)} \top$. Hence, $\rho,i\models \alpha(\eta_{root})$ iff $i=1=pos_\rho(\eta_1)$.
\item If $\eta = EP \eta_1$ then $\alpha(\eta_1) = \neg F_{(0,\infty)} \top$. Hence, $\rho,i\models \alpha(\eta_{root})$ iff $i=\#\rho=pos_\rho(\eta_1)$.
\item If $\eta$ is of the form $\eta_1 \lor \eta_2$ or $\neg \eta$ or $x.\eta_1$ then $\alpha(\eta_1) = \alpha(\eta)$. This follows from the fact that $pos_\rho(\eta)= pos_\rho(\eta_1) = pos_\rho(\eta_2)$.
\item Now consider the main case of $\eta= X_\theta \eta_1$ with $\theta=(a,g)$. For given $\theta$, we define a corresponding \mitlfp\/ formula $\mathcal{CF}(\theta,\eta)$
such that the following proposition holds:
\begin{proposition}\label{prop:guard}
$\rho,i,val_\rho(\eta)\models\theta~$ iff $~\rho,i\models \mathcal{CF}(\theta,\eta)$.
\end{proposition}
Using this we define $\alpha(\eta_1)$ and show that $\rho,i\models \alpha(\eta_1)$ iff $i=pos_\rho(\eta_1)$.
\item The case of $\eta= Y_\theta \eta_1$ is symmetric to the above case.
\end{itemize}
Given $\theta=(a,g)$, define $\mathcal {CF}(\theta,\eta) ~=~ a \land  {\mathcal C}(g,\eta)$ 
where the construction of $\mathcal C(g,\eta)$ is given in Table \ref{table:constraint}. Note that any constraint of the form $x -T =c$ can be replaced by equivalent constraint $(x - T \leq c) \land (x-T \geq c)$. Similarly, for $T - x =c$ too. We omit from Table \ref{table:constraint}, the remaining cases of $T -x \approx c$ which are similar.
\begin{table}
\begin{center}
\begin{tabular}{|l|l|}
\hline
$~g~$ & $~{\mathcal C}(g,\eta)~$ \\
 \hline
 $~x - T < c~$ & $~F_{[0,c)} \alpha(anc_x(\eta))~$ \\
 \hline
 $~x - T \leq c~$ & $~F_{[0,c]} \alpha(anc_x(\eta))~$ \\
 \hline
 $~x - T > c~$ & $~F_{(c,\infty)} \alpha(anc_x(\eta))~$ \\
 \hline
  $~x - T \geq c~$ & $~F_{[c,\infty)} \alpha(anc_x(\eta))~$ \\
 \hline
 $~T - x < c~$ &  $~P_{[0,c)} \alpha(anc_x(\eta))~$ \\
 \hline
 $~g_1 \land g_2~$ & $~{\mathcal C}(g_1,\eta) \land  {\mathcal C}(g_2,\eta)~$ \\
 \hline
\end{tabular}
\end{center}
\caption{}
\label{table:constraint} 
\end{table}
To sketch the proof of proposition \ref{prop:guard}, we first show that
$\rho,i\models \mathcal C(g,\eta)$ iff $val_\rho(\eta),\tau_i\models g$.
From this, it follows that $\rho,i,val_\rho(\eta)\models\theta$ iff $\rho,i\models \mathcal {CF}(\theta,\eta)$.
Consider the case where $g = x-T<c$. Then, $\mathcal C(g,\eta) = F_{[0,c)} \alpha(anc_x(\eta))$. By semantics of \mitlfp, we know that $\rho,i\models \mathcal C(g,\eta)$ iff $\exists j>i$ such that (i) $j=pos_\rho(anc_x(\eta))$ (using the inductive hypothesis) and (ii) $\tau_j-\tau_i \in [0,c)$. However, from proposition \ref{prop:anc}, we know that $val_\rho(\eta)(x) = \tau_j$. Hence,
(i) and (ii) hold iff $val_\rho(\eta),\tau_i\models g$. The other cases may be proved similarly.

Now, define $\alpha(\eta_1) = \mathcal {CF}(\theta,\eta) ~\land~ (P_{(0,\infty)} \alpha(\eta)) ~\land~ (\neg P_{(0,\infty)} (\mathcal {CF}(\theta,\eta) \land P_{(0,\infty)} \alpha(\eta) ))$.\\
The three conjuncts of the above formula respectively give the following observations. $\rho,i\models \alpha(\eta_1)$ iff
(i) $\rho,i,val_\rho(\eta_1)\models \theta$ (from proposition \ref{prop:guard}),
(ii) $\exists j<i ~.~ j=pos_\rho(\eta)$ (from induction hypothesis), and 
(iii) $\forall k ~.~ pos_\rho(\eta) <k<i ~.~ \rho,k,val_\rho(\eta_1)\not\models \theta$
\qed
\end{proof}

Now, define the evaluation $eval_\rho(\eta)$ of a subformula as its truth value at its 
deterministic position $pos_\rho(\eta)$. This can be defined as follows:
If $pos_\rho(\eta) \not= \bot$ then
$eval_\rho(\eta) = (\rho,val_\rho(\eta),pos_\rho(\eta) \models \eta)$ and $\mathit{false}$
otherwise. Clearly,
since $pos_\rho(\eta_{root})=1$ and $val_\rho(\eta_{root})= \nu_0$,
it follows that $eval_\rho(\eta_{root}) = ((\rho,1,\nu_0) \models \eta_{root})$. 
\begin{theorem}
\label{lem:ttlbeta}
For every subformula $\eta$, we construct an \mitlfp\/ formula $\beta(\eta)$ such that
$eval_\rho(\eta) \fif pos_\rho(\eta)\neq \bot$ and $\rho,pos_\rho(\eta) \models \beta(\eta)$.
\end{theorem}
The construction of $\beta(\eta)$ is by induction on the structure of $\eta$. In its construction, we use the formula $\alpha(\eta)$ given earlier.
If $\eta=\top$ then $\beta(\eta) = \alpha(\eta)$. 
If $\eta = \theta$ then $\beta(\eta) = \alpha(\eta) \land {\mathcal CF}(\theta,\eta)$.
If $\eta = \eta_1 \lor \eta_2$ then $\beta(\eta) = \alpha(\eta) \land 
(\beta(\eta_1) \lor \beta(\eta_2))$. If $\eta = x.\eta_1$ then $\beta(\eta) =
\beta(\eta_1)$. If $\eta = \neg \eta_1$ then $\beta(\eta) = \alpha(\eta) \land \neg \beta(\eta_1)$. Now, we consider the main case.
Let $\eta= X_\theta \eta_1$. Then, $\beta(\eta) = \alpha(\eta) \land \fut(\alpha(\eta_1) \land \beta(\eta_1))$. It is easy to prove by induction on the height of $\eta$ that Theorem \ref{lem:ttlbeta} holds.

\subsection{On limited expressive power of \ttl}
Given any \ttl\/ formula, its \textit{modal depth} corresponds to the maximum number of modal operators in  any path of its parse tree and its \textit{modal count} corresponds to the total number of modal operators in the the formula.\\
A \ttl\/ formula $\phi$ is said to \textit{reach} a position $i$ in a word $w$, if there exists a subformula $\eta$ of $\psi$ such that $\pos(\eta)=i$. 

\begin{theorem}
\begin{enumerate}
\item \bmtlfp $\nsubseteq$ \ttl
\item \mitlfp $\nsubseteq$ \ttl
\end{enumerate}
\end{theorem}
\begin{proof}
 
(i) Consider the \bmtlfp\/ formula $\phi:= \fut_{(0,1)}(a\land \fut_{[3,3]} c)$ given in the proof  of Theorem \ref{theo:mtlfrag2} and $A_n$ and $B_n$ be as defined in that proof. Let 
$w_n=A_{n+1}$ and $v_n=B_{n+1}$. Thus, both $w_n$ and $v_n$ consist of events $a^{2n+3}c^{2n+3}$. Then, $\forall n ~.~ w_n\not\in\mathcal L(\phi)$ and $v_n\in \mathcal L(\phi)$. 
\begin{proposition}\label{prop:ttl1}
 For $n>1$, no \ttl\/ formula of modal depth $1 \leq m\leq n$ can reach the middle $2n-2m+3$ $a$'s or the middle $2n-2m+3$ $c$'s in $w_n$.
\end{proposition}
\begin{proof}
Firstly, note that if no \ttl\/ formula of depth $m$ can reach a position $i$ in a word, then its boolean combinations also cannot reach $i$ in the word. We now prove the claim by induction on $m$, for some fixed $n$. \emph{Base step:} $m=1$ : Since all $a$ satisfy the same integral guards and
all $c$ also satisfy the same set of intergral guards, the topmost modality may match either the 
first or last $a$ or the first or last  $c$ in $w_n$ (irrespective of the guard that is chosen). Hence, the middle $2n-2+3$ $a$'s and $c$'s cannot be reached. \emph{Induction Step:} Let the proposition be true for some $1 \leq m<n$. Hence for every $\psi$ of modal depth $m$, $\psi$ cannot reach the middle $(2n-2m+3)$ $a$'s and $c$'s in $w_{n}$. Every \ttl\/ formula $\psi'$ of modal depth $m+1$ may be obtained from some \ttl\/ formula $\psi$ of modal depth $m$, by extending every path in parse tree of $\psi$ by at most one modality in the end. However, since all the middle $2n-2m+3$ $a$'s and $c$'s satisfy the same integral guards with respect to the time stamps of the peripheral $a$'s and $c$'s, adding another modality to $\psi$ can make $\psi'$ reach at most the $(m+1)^{th}$ or $n-(m+1)^{th}$ $a$ or $c$. This leaves us with $2n-2m+3-2 ~=~ 2n-2(m+1)+3$ middle $a$'s and $c$'s which remain unreachable.
\qed
\end{proof}
Consider a \ttl\/formula of modal depth $\leq n$. From proposition \ref{prop:ttl1},  the middle 3 $a$'s and $c$'s are unreachable. Moreover, they satisfy the same set of time constraints with respect to the reachable events. Hence, perturbing the middle $c$ alone  will not change the truth of the formula as $c$ it will continue to satisfy the same set of timing constraints w.r.t. the reachable events. Hence $w_n\models \psi$ iff $v_n\models\psi$. Since $w_n\not\models\phi$ and $v_n\models\phi$, no $\psi$ of modal depth $\leq n$ can distinguish between $w_n$ and $v_n$. Hence, we can conclude that there is no \ttl\/ formula equivalent to $\phi$.\\
  \\
(ii)Consider the \mitlfp\/ formula $\phi:= \fut_{[0,\infty)}(a\land \fut_{(1,2)} c)$ and let $\psi$ be a \ttl\/ formula of modal count $m$ such that $\mathcal L(\phi)=\mathcal L(\psi)$. Assuming that freeze variables in $\psi$ are not reused, there are a maximum number of $m$ freeze variables in $\psi$. 
Now consider the word $w$  consisting of event sequence $(ac)^{4m+1}$ where the $x$'th $ac$ pair gives the timed subword $(a,2x)(c,2x+0.5)$. Thus, each $c$ is 0.5 t.u. away from its paired $a$, and $2.5$ units away from the $a$ of the previous pair. Hence, $w\not\in\mathcal L(\phi)$. \\
Consider the evaluation of $\psi$ over $w$. Each of the $m$ freeze variables is frozen at most once, in the evaluation of $\psi$. By a counting argument, there are at least $m+1$ (possibly overlapping but distinct) subwords of the form $acacac$, none of whose elements are frozen. Call each such subword a group. Enumerate the groups sequentially.
Let $v_j$ be a word identical to $w$ except that the $j^{th}$ group is altered, such that its middle $c$ is shifted by 0.7 t.u. to the right, so that $v_j$ satisfies the property $\phi$. Note that there are at least $m+1$ such distinct $v_j$'s and for all $j$, $v_j\in\mathcal L(\phi)$. \\
\emph{Claim:} Given a $v_j$, if there exists a subformula $\eta$ of $\psi$ such that $Pos_{v_j}(\eta)$ matches the altered $c$, then for all $k\neq j$, $Pos_{v_k}(\eta)$ does not match its altered $c$. (This is because, the altered $c$ in $v_j$ must satisfy a guard which none of its two surrounding $c$'s in the group can satisfy). \\
From the above claim, we know that the $m$ modalities in $\psi$, may match its position in at most $m$ of the altered words $v_j$. However, the family $\{v_j\}$ has at least $m+1$ members. Hence, there exists a $k$ such that the altered $c$ of $v_k$, (and the $k^{th}$ group) is not reachable 
by $\psi$ in $w$ or any of the $\{v_j\}$. Hence $w\models\psi$ iff $v_k\models \psi$. \\
Therefore, there is no \ttl\/ formula which can express the language $\mathcal L(\phi)$.
\qed
\oomit{  
(iii)Consider the language defined by the \bmitlus\/ formula $\phi := \top \until_{(1,2)}[b\land \{b\until_{(0,1)}(c\land (c\until_{(0,1)}d))\}]$.
In the proof of Theorem \ref{theo:mtlfrags} in Appendix B, we show that $L(\phi)$ is not expressible 
in \mtlfp.  Also, by Theorem \ref{lem:ttlbeta}, \ttl\/ is a sub logic of \mtlfp. Hence, we infer that the language defined by $\phi$ is not expressible in \ttl.
}
\end{proof}


\newpage
\appendix


\section{$\INTV$ \mtlus\/ EF theorem}
\label{app:efthmproof}

\begin{lemma}
\label{lem:wordlenboundform}
For any $\phi \in \mtlus$ and any integer $n$, let $[\phi]^{n}$ denote the formula obtained by replacing in $\phi$ any occurrence 
of any constant $c> n$  (or $\infty$) by $n$. Let $\rho$ be a timed word and let integer $k > \tau_{\#\rho}$.
Thus, $k$ is an integer strictly larger than the last time stamp in $\rho$.
Then, $\forall i$ . 
 $\rho,i \models \phi \iff \rho,i \models [\phi]^k$. \qed
\end{lemma}

\noindent{\bf Theorem \ref{thm:ef}}
 $(\rho_0,i_0) \gameeq{k}^{\intv} (\rho_1,i_1)$ if and only if $(\rho_0,i_0) \formeq{k}^{\INTV} (\rho_1,i_1)$.
 
\begin{proof}
Let $\gameeq{k}$ and $\formeq{k}$ denote $\gameeq{k}^{\INTV}$ and $\formeq{k}^{\INTV}$, respectively.
The proof is by induction on $k$. For $k=0$ the result is immediate; since for any 0-modal depth atomic formula $a \in \Sigma$, \ddp\/ wins the 0-round game iff $\sigma_0[i_0]=\sigma_1[i_1]$. As induction step, we now prove the result for $k+1$ assuming that it holds for $k$.

($\Rightarrow$) Assume that $(\rho_0,i_0) \gameeq{k+1} (\rho_1,i_1)$. We consider the representative case of  $\rho_0,i_0 \models \phi$ for a $k+1$-modal depth formula whose topmost operator is $\until_I$. Our aim is to show that $\rho_1,i_1 \models \phi$.
Similar argument holds when the topmost operator is $\since_I$ or the role of the two words is reversed. Hence, the theorem holds for the boolean combinations such formulas completing the proof.

Since $\rho_0,i_0 \models \psi \until_I \gamma$ there exists $i'_0>i_0$ such that $\rho_0,i'_0 \models \gamma$ and $\forall i''_0: i_0<i''_0 < i'_0.~ \rho_0,i''_0 \models \psi$. Now, one possible play of \ssp\/ is to choose $\delta=0$ and $\until_I$ move with position $i'_0$, followed by the $F$-part. The \ddp\/ has a winning strategy by which he can choose position $i'_1>i_1$ in $\rho_1$ such that $\rho_0,i'_0 \gameeq{k} \rho_1,i'_1$. By induction hypothesis, we have $\rho_0,i'_0 \formeq{k} \rho_1,i'_1$ and since $\rho_0,i'_0 \models \gamma$ for the $k$-modal depth formula $\gamma$, it follows from formula equivalence that $\rho_1,i'_1 \models \gamma$. Another play of $\ssp$ is to choose $U$-part in above giving position $i''_1$ to which \ddp\/ can respond with winning strategy by choosing
$i''_0$ such that $\rho_0,i''_0 \gameeq{k} \rho_1,i''_1$. By induction hypothesis,
we have $\rho_0,i''_0 \formeq{k} \rho_1,i''_1$ and since
$\rho_0,i''_0 \models \psi$ for the $k$-modal depth formula $\psi$, it follows from formula equivalence that $\rho_1,i''_1 \models \psi$.

$(\Leftarrow$) Assume that $(\rho_0,i_0) \not\gameeq{k+1} (\rho_1,i_1)$. We must find a $k+1$-modal depth formula distinguishing the two structures. 
The choice of the formula depends upon the play. We consider the interesting case by which \ssp\/ first plays an
$\until_I$ move with  word $\delta$ and wins the $k+1$-round game. Let the position chosen by \ssp\/ in Part I be $i'_\delta$. 
Let $m= \lfloor max(\tau_0[\#\rho_0],\tau_1[\#\rho_1]) \rfloor +1$, i.e. $m$ is the smallest integer strictly greater than the last time stamps of $\rho_0$ and $\rho_1$. For 
$j \in dom( \rho_\delta)$, let $\phi^{k,m}_j$ be conjunction of all $\mtlus$ formulas $\zeta$ with $\maxint(\zeta) \leq m$ and modal depth $k$ such that $\rho_\delta,j \models \zeta$. Note that, up to equivalence, there are only finitely many such formulas $\zeta$ and the conjunction can be written as an $\mtlus$ formula. The key property which follows from
Lemma \ref{lem:wordlenboundform} is that if 
$\rho_\deltabar,i \models \phi^{k,m}_j$ then $(\rho_\delta,j)$ and $(\rho_\deltabar,i)$
satisfy the same set of $k$-modal depth \mtlus\/ formulas.

Now consider the formula
\[
\psi = (\phi^{k,m}_{i_\delta + 1} \lor \cdots \lor \phi^{k,m}_{i'_\delta - 1}) 
~\until_I~ (\phi^{k,m}_{i'_\delta}) 
\]
Then, $\psi$ has modal depth $k+1$. For simplicity denote $\psi$ as $P \until_I Q$.
By construction it is clear that $\rho_\delta,i_\delta \models \psi$. Now, we claim that if $\rho_\deltabar, i_\deltabar \models \psi$ then $\ddp$ would have won the game in which \ssp\/ made the Part-I $\until_I$ move at position $\rho_\delta,i'_\delta$. This is a contradiction
and hence $\rho_\deltabar, i_\deltabar \not \models \psi$. To see the claim, let $i'_\deltabar > i_\deltabar$ be the position such that $\rho_\deltabar,i'_\deltabar \models Q$. 
The \ddp\/ would respond to the Part-I move by choosing $i'_\deltabar$. Now, by definition of $Q$, $(\rho_\deltabar, i'_\deltabar)$ satisfies same set of $k$-modal depth formulas as satisfied by $(\rho_\delta,i'_\delta)$. Hence, $(\rho_0,i'_0) \formeq{k} (\rho_1,i'_1)$. By inductive hypothesis, then $(\rho_0,i'_0) \gameeq{k} (\rho_1,i'_1)$ and the \ddp\/ can force a win from this configuration. Hence, \ddp\/ would win if \ssp\/ chose $F$-part in first round.
Assume that \ssp\/ chooses the $U$-part and chooses a position $i''_\deltabar$ s.t.
$i_\deltabar < i''_\deltabar < i'_\deltabar$. Clearly, $\rho_\deltabar,i''_\deltabar \models P$ and hence $\rho_\deltabar,i''_\deltabar \models \phi^{k,m}_l$ for some $l$ with $i_\delta < l < i'_\delta$. The \ddp\/ responds to $U$-part by choosing $i''_\delta = l$.
Clearly, $(\rho_\deltabar,i''_\deltabar)$ and $(\rho_\delta,i''_\delta)$ satisfy the same set of $k$-modal depth formulas. Hence they are $\formeq{k}$ and ,using Induction Hypothesis,
$\gameeq{k}$. Hence, the \ddp\/ can force a win from these configurations.
\qed
\end{proof}

\section{Separating sub logics of \mtlus\/ and \tptlus}
\label{app:mtlus}

For timed logics, weakly monotonic time includes \emph{instantaneous} timed words where all its letters occur at the initial time point $0$. 
Logic $\bmtlus^0$ denotes $\mtlus$ where all the $\until_I$ and $\since_I$ and modalities only have the interval $I=[0,0]$.
We have the following property.
\begin{proposition}
 Over instantaneous timed words, (a) $\mtlus \equiv \bmtlus^0$, 
 (b) $\mtlfp \equiv \bmtlfp^0$, (c) $\tptlus \equiv \bmtlus^0$, and (d) $\tptlfp \equiv \bmtlfp^0$.
\end{proposition}
To see these, notice that over instantaneous timed words the modalities $\until_{[0,i\rangle}$ are all equivalent for $i \geq 0$ (incl. $\infty$). Moreover, modalities
$\until_{\langle i,j \rangle}$ with $i>0$ all evaluate to ``false''. Hence, we can reduce every $\mtlus$ formula to $\bmtlus^0$ formula (of same modal depth) which uses only the interval $[0,0]$. In an analogous fashion it is also easy to reduce a $\tptlus$ formula to equivalent $\bmtlus^0$
formula using only the interval $[0,0]$.
Note that every freeze quantification over instantaneous word trivially sets a variable to $0$.
Hence, a guard $g$ can be replaced by its truth value found by setting  all variables to $0$. Freeze quantification can be omitted and every $\until$ or $\since$ can be replaced by $\until_{[0,0]}$ or  $\since_{[0,0]}$.
Moreover, unary modalities are preserved by these reductions.

Now, logic $\bmtlus^0$ over instantaneous words is semantically isomorphic to $\ltlbin$ over the corresponding untimed word. Hence, all the $\ltl$ separations carry over to logics MTL and TPTL over weakly monotonic time (which include instantaneous words).
Etessmai and Wilke established an until hierarchy within LTL formulas.
Specifically, they have shown that $LTL[\until,\since] \not \subseteq LTL[\fut,\past]$. From
this, and above argument, we immediately conclude that
\begin{theorem} Over weakly monotonic timed words,
 \begin{itemize}
  \item $\bmtlus \not \subseteq \mtlfp$ 
  \item $\bmtlus \not \subseteq \tptlfp$.
 \end{itemize}
\end{theorem}

We now consider strictly monotonic time. A timed word is called \emph{unitary} if all its letters occur within the open interval $(0,1)$ and with distinct time stamps (i.e. it is strictly monotonic).
Let $\bmtlus^{(0,1)}$ denote $\mtlus$ where the only interval used in the modalities is $I=(0,1)$. This is a subset of $\bmtlus^1$.
\begin{proposition}
Over unitary timed words, $\mtlus \equiv \bmtlus^{(0,1)}$ and $\mtlfp \equiv \bmtlfp^{(0,1)}$ 
\end{proposition}
To see this, over unitary timed words every modality $\until_I$ is equivalent to $\until_{(0,1)}$ if $(0,1) \subseteq I$, and equivalent to ``false'' otherwise. Moreover this preserves unary modalities. 

Now, logic $\bmtlus^{(0,1)}$ over unitary timed words is semantically isomorphic to $\ltlbin$ over the corresponding untimed word. Hence, all the $\ltl$ separations carry over to the logic MTL over strictly monotonic timed words (which include unitary words). Specifically, as  we have $LTL[\until,\since] \not \subseteq LTL[\fut,\past]$, we conclude the following.
\begin{theorem}
 Over strictly monotonic timed words, $\bmtlus 
 \not \subseteq \mtlfp$.
\end{theorem}

\subsection{Theorem \ref{theo:mtlfrag2}: Strategy of the Duplicator}
As stated before, we consider $n$ round game on pair of words $A_{2n}$ and $B_{2n}$ 
described earlier. (Recall these definitions.) Both words are identical except for the time stamp of the middle $c$. A game configuration $(i,j)$ is written by putting $*$ before $i^{th}$ letter in the first word $A_{2n}$ and $j$'th letter in the second word $B_{2n}$. The initial configuration is $(* a^{4n+1}  c^{4n+1}, ~*a^{4n+1}  c^{4n+1})$. 
\begin{itemize}
 \item As the two words are identical except the time stamp of the middle (i.e. 2n+1th) $c$, the strategy of \ddp\/ is to play a configuration of the form $(i,i)$ whenever possible. (Duplicator wins the round if she can achieve such a configuration.) Such a configuration $(i,i)$ is called an identical configuration.
 \item The optimal strategy of \ssp\/ is to get out of identical configurations as quickly as possible. \ssp\/ can play middle $a$ in $B_{2n}$ to which \ddp\/ responds in copycat fashion
 giving configuration $(2n+1,2n+1)$. Next, \ssp\/ can move to middle $c$ in $A_{2n}$ with an 
 $\until_I$ move. Note that distance between middle $a$ and middle $c$ in $A_{2n}$ is 
 $3+\epsilon$ whereas distance between middle $a$ and middle $c$ is $3$  in $B_{2n}$. 
 The move of \ddp\/ depends upon the interval chosen by \ssp. If \ssp\/ chooses an 
 interval $I= [3,j \rangle$, the \ddp\/ responds by moving pebble to middle ``c'' giving an identical configuration. A better move of \ssp\/ is to
 choose interval $I=(3,j\rangle)$. \ddp\/ is forced to place her pebble somewhere after the middle $c$. The optimal move of \ddp\/ is to place pebble at position next to the middle $c$ giving the non-identical configuration $(a^{4n+1}c^{2n}*cc^{2n},~a^{4n+1}c^{2n}c*c^{2n})$. In this configuration, the distance of the pebbles from the right end differs by 1. (We remark that there are several alternative plays of 2 moves where \ssp\/ can achieve a non-identical configuration where the distance of the pebbles from the nearer end of the block of $a$ or $c$ differs by 1. \ssp\/ cannot do better than this.)
 \item The optimal strategy of spoiler from the above non-identical configuration is  to play a sequence of ``until'' moves placing $A_{2n}$ pebble 2 positions to right in each move. (If \ssp\/ moves more than 2 positions the \ddp\/ can immediately achieve an identical configuration. Moving less than 2 position makes the game last longer.) The optimal move of \ddp\/ is to do the same by moving $B_{2n}$ pebble 2 positions to the right. (Note that due to the conditions of until move, the \ddp\/ is forced to move his pebble two or more positions to the right.) Such a move results in a non-identical configuration where pebble distance from the right end decreases. The \ddp\/ can sustain this for $n-1$ ``until'' moves of \ssp\/ finally giving the configuration 
 $(a^{4n+1}  c^{4n-2}  * ccc, ~a^{4n+1}  c^{4n-1}  * cc)$. The next ``until'' move of the spoiler 2 positions to right in $A_{2n}$ cannot be duplicated. \qed 
\end{itemize}

\end{document}

%% file: macros.tex
\newcommand{\fif}{\rmiff}

\newcommand{\rmiff}{\mathbin{~\mbox{iff}~}}

\newcommand{\auth}{\textsf}
\newcommand{\book}{\textit}

\newcommand{\jour}{\textsl}
\newcommand{\vol}{}

\newcommand{\pub}[2]{ (#1, #2)}
\newcommand{\yr}[1]{, #1}
\newcommand{\pp}[2]{, #1--#2}

\newcommand{\defn}{\mathrel{\mbox{$~\stackrel{\rm def}{=}~$}}}

\newcommand{\opr}{\mbox{$\mathit{Opr}$}}
\newcommand{\sub}{\mbox{$\mathit{Subf}$}}
\newcommand{\intv}{\mbox{$\mathit{Intv}_w$}}

\newcommand{\pos}{\mbox{$\mathit{Pos}_w$}}

\newcommand{\until}{\textsf{U}}
\newcommand{\since}{\textsf{S}}
\newcommand{\fut}{\textsf{F}}
\newcommand{\past}{\textsf{P}}

\newcommand{\ltl}{\mbox{$\mathit{LTL}$}}
\newcommand{\ltlbin}{\mbox{$\ltl[\until,\since]$\/}}

\newcommand{\oomit}[1]{}
\newcommand{\df}{\defn}

\newcommand{\ap}{\mbox{$\mathit{AP}$}}

\newcommand{\mtlus}{\mbox{$\mathit{MTL[\until_I,\since_I]}$}}
\newcommand{\mtlfp}{\mbox{$\mathit{MTL[\fut_I,\past_I]}$}}
\newcommand{\mitlus}{\mbox{$\mathit{MITL[\until_I,\since_I]}$}}
\newcommand{\bmtlus}{\mbox{$\mathit{Bounded MTL[\until_I,\since_I]}$}}
\newcommand{\bmtlfp}{\mbox{$\mathit{Bounded MTL[\fut_I,\past_I]}$}}
\newcommand{\mitlfp}{\mbox{$\mathit{MITL[\fut_I,\past_I]}$}}
\newcommand{\bmitlus}{\mbox{$\mathit{Bounded MITL[\until_I,\since_I]}$}}
\newcommand{\bmitlfp}{\mbox{$\mathit{Bounded MITL[\fut_I,\past_I]}$}}
\newcommand{\tptlus}{\mbox{$\mathit{TPTL[\until,\since]}$}}
\newcommand{\tptlfp}{\mbox{$\mathit{TPTL[\fut,\past]}$}}
\newcommand{\tptlf}{\mbox{$\mathit{TPTL[\fut]}$}}

\newcommand{\potdta}{\mbox{$\mathit{po2DTA}$}}

\newcommand{\ttl}{\mbox{$\mathit{TTL[X_\theta, Y_\theta]}$}}

\newcommand{\ssp}{\mbox{$\mathit{Spoiler}$}}
\newcommand{\ddp}{\mbox{$\mathit{Duplicator}$}}
\newcommand{\aaa}{\mathcal A}
\newcommand{\bbb}{\mathcal B}

\newcommand{\gameeq}[1]{\approx_{#1}}
\newcommand{\formeq}[1]{\equiv_{#1}}

\newcommand{\maxint}{\mbox{$\mathit{MaxInt}$}}

\newcommand{\ttlxy}{\mbox{$\mathit{TTL}[X_\theta,Y_\theta]$}}
\newcommand{\mtlu}{\mbox{$\mathit{MTL[\until_I]}$}}

\newcommand{\tptlu}{\mbox{$\mathit{TPTL[\until]}$}}

\newcommand{\rlpos}{{\mathbb{R}_{0}}}

\newcommand{\INTV}{\mbox{$\mathit{Iv}$}}
\newcommand{\zintv}{\mathbb{Z}I}
\newcommand{\extintv}{\mathbb{Z}IExt}
\newcommand{\bintv}{Bd\mathbb{Z}I}
\newcommand{\bextintv}{Bd\mathbb{Z}IExt}
\newcommand{\deltabar}{{\overline{\delta}}}